\documentclass[letterpaper, 10 pt, conference]{ieeeconf}
\usepackage{graphics, subfig}
\graphicspath{{./images/}}
\usepackage{epsfig,cite}
\usepackage{mathptmx,mathtools}
\usepackage{times}

\usepackage{algorithm} 
\usepackage{cite}
\usepackage{algorithmic}
\usepackage[algo2e,ruled,vlined]{algorithm2e}
\usepackage{amsmath,amsthm,mathtools}
\usepackage{amssymb,soul} 
\usepackage{algorithmic,algorithm}
\usepackage{relsize}
\usepackage{comment}
\newtheorem{theorem}{Theorem}
\newtheorem{lemma}{Lemma}

\theoremstyle{definition}
\newtheorem{definition}{Definition}

\theoremstyle{remark}

\usepackage{xcolor}
\IEEEoverridecommandlockouts                          
\title{\LARGE \bf
On Determining and Qualifying the Number of Superstates in Aggregation of Markov Chains}
\author{Amber Srivastava, Raj K. Velicheti, and Srinivasa M. Salapaka% <-this % stops a space
\thanks{This work was supported by  NSF grant ECCS (NRI) 18-30639. The authors are with the Mechanical Science and Engineering Department and Coordinated Science Laboratory, University of Illinois at Urbana-Champaign, IL, 61801 USA. E-mail: \{asrvstv6, rkv4, salapaka\}@illinois.edu. Paper is under consideration at Pattern Recognition Letters}
}
\begin{document}
\maketitle
%%%%%%%%%%%%%%%%%%%%%%%%%%%%%%%%%%%%%%%%%%%%%%%%%%%%%%%%%%%%%%%%%%%%%%%%%%%%%%%%
\begin{abstract}
Many studies involving large Markov chains require determining a smaller representative (aggregated) chains. Each {\em superstate} in the representative chain represents a {\em group of related} states in the original Markov chain. Typically, the choice of number of superstates in the aggregated chain is ambiguous, and based on the limited prior know-how. In this paper we present a structured methodology of determining the best candidate for the number of superstates. We achieve this by comparing aggregated chains of different sizes. To facilitate this comparison we develop and quantify a notion of {\em marginal return}. Our notion captures the decrease in the {\em heterogeneity} within the group of the {\em related} states (i.e., states represented by the same superstate) upon a unit increase in the number of superstates in the aggregated chain. We use Maximum Entropy Principle to justify the notion of marginal return, as well as our quantification of heterogeneity. Through simulations on synthetic Markov chains, where the number of superstates are known apriori, we show that the aggregated chain with the largest marginal return identifies this number. In case of Markov chains that model real-life scenarios we show that the aggregated model with the largest marginal return identifies an inherent structure unique to the scenario being modelled; thus, substantiating on the efficacy of our proposed methodology.
\end{abstract}
%%%%%%%%%%%%%%%%%%%%%%%%%%%%%%%%%%%%%%%%%%%%%%%%%%%%%%%%%%%%%%%%%%%%%%%%%%%%%%%%
\section{Introduction}\label{sec: Introduction}
Markov chains provide a mathematical model to study many real-world stochastic processes, such as population dynamics, cruise control, transportation systems, and queueing networks \cite{gagniuc2017markov}. However, Markov chain models for several complex systems such as applications originating in network analysis \cite{srikant2004mathematics}, neuroscience \cite{quinn2011estimating}, and economics \cite{zhang2004nearly} require large number of states where analyzing them is challenging and inefficient; thus, laying foundation for several Markov chain aggregation techniques.

The two-fold objective of the aggregation problem is to (a) group {\em similar} states in the Markov chain and represent them as a single {\em superstate} in the aggregated model, and (b) determine the state transition probability matrix of the aggregated model. Prior works in literature \cite{aldhaheri1991aggregation,rached2004kullback,beck2009model,vidyasagar2012metric,xu2013distance,deng2011optimal} aim at developing an appropriate {\em dissimilarity} metric to compare the differently-sized state transition probability matrices of the Markov chain and its aggregated model. The algorithms proposed in \cite{xu2014aggregation,geiger2014markov,deng2011optimal} determine the aggregated models at a {\em pre-specified} number of superstates such that the dissimilarity is minimized. One of the recent works \cite{sledge2019information} considers the aggregation of the class of nearly completely decomposable (NCD) \cite{ando1963near} Markov chains, and provides an algorithm-based method to estimate the number of superstates underlying the above class of chains. There is scant literature that provides algorithm-free methodology on the choice of number of superstates for general Markov chain. Much of this is owing to the unavailability of a method to quantitatively compare the aggregated models of different sizes that describe the same Markov chain. In this article, we provide a procedure to compare these aggregated models. We demonstrate its utility in determining an appropriate choice for the number of superstates among the {\em given} aggregated models. More precisely, we develop and quantify a notion of {\em marginal return} that compares the aggregated models of different sizes. We demonstrate via simulations that the aggregated model with the largest marginal return estimates the superstates {\em underlying} the original Markov chain. The corresponding size of this aggregated model provides an appropriate choice for the number of superstates.

\begin{figure}
    \centering
    \includegraphics[width=0.85\columnwidth]{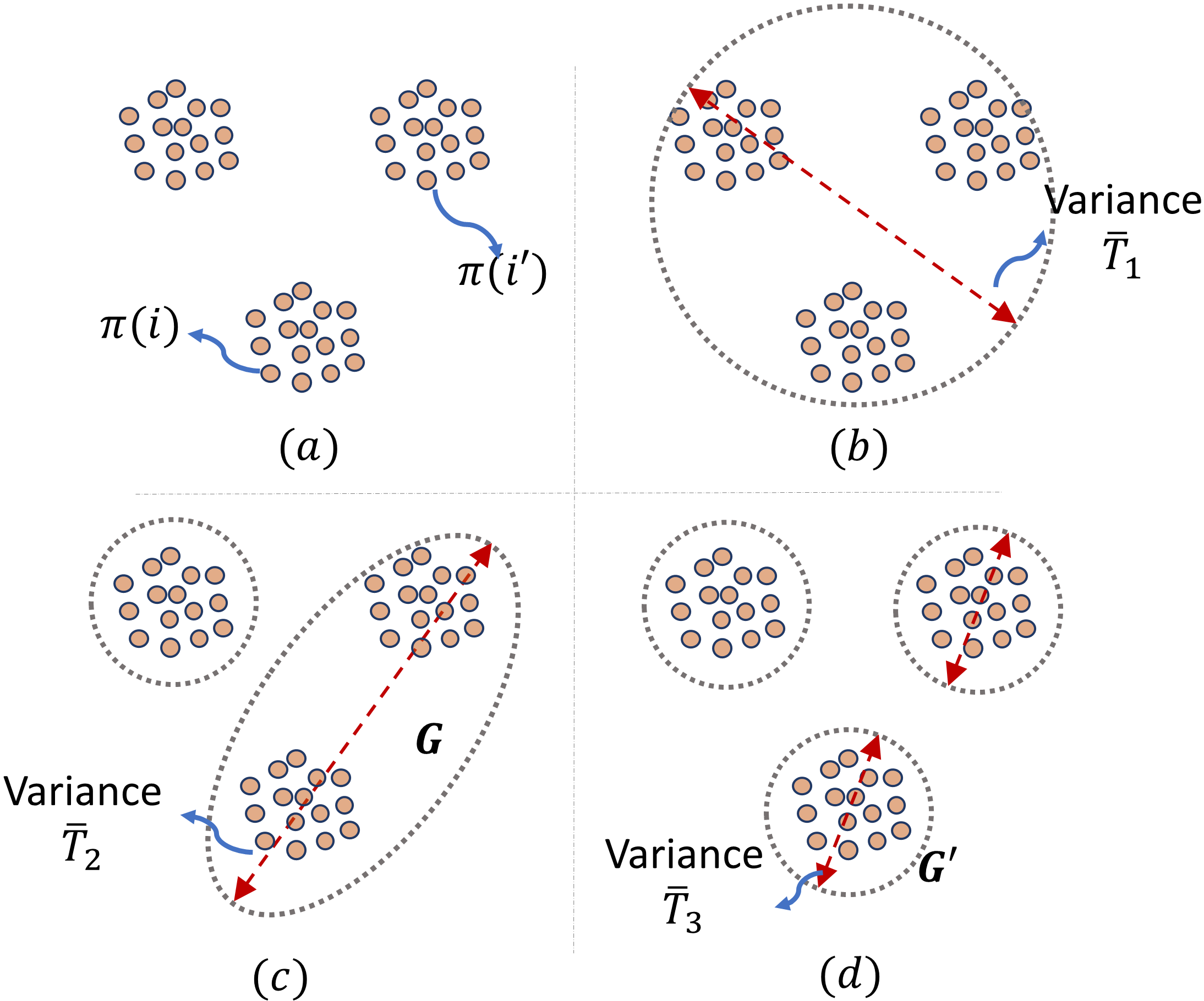}
    \caption{(a) Represents the states in the original Markov chain. (b) Aggregated model with 1 superstates, i.e. all states in the same group. (c) Illustrates the 2 superstates, i.e., 2 groups of related states. (c) Illustrates the 3 distinct groups of related states.}\label{fig: IntroFig}\vspace{-0.4cm}
\end{figure}

The primary goal in the aggregation of Markov chain is akin to the  data-clustering problem. The former identifies groups of {\em similar} states, whereas the latter determines the groups (or, clusters) of {\em similar} data points. Thus, one can view the problem of determining the appropriate number of superstates to be analogous to estimating the {\em true} (or, natural) number of clusters in a dataset. The latter is a well-studied  problem in the data clustering literature \cite{tibshirani2001estimating,feng2007pg,pelleg2000x,sugar2003finding,hamerly2004learning,kalogeratos2012dip,srivastava2019persistence}. However, the methods proposed in the context of data clustering are not directly applicable to the aggregation of Markov chains. This is owing to the inherent differences in the two problems. The data clustering problem identifies groups (or, clusters) of similar data points lying in the euclidean space where the cost function is typically 2-norm. On the other hand, the Markov chain aggregation requires determining the groups of similar states whose the transition probabilities lie in the probability space, and the cost is typically in terms of the relative entropy of the transition probabilities. Nonetheless, the methods proposed in the context of data clustering can still be helpful in devising {\em procedures} to compare aggregated chains of different sizes. These procedures can then be used to estimate the appropriate number of superstates. In fact, our paper builds on the prior work done in \cite{srivastava2019persistence} that develops and quantifies a notion of {\em persistence} of a clustering solution, and utilizes it to estimate the natural number of clusters in the dataset. Below we provide a qualitative and quantitative descriptions of our notions of marginal return and heterogeneity.

{\bf Qualitative description of marginal return: } As stated above we exploit the work done in \cite{srivastava2019persistence} to qualitatively elucidate our notion of marginal return and the heterogeneity of a superstate. We first give a brief description of the idea presented in \cite{srivastava2019persistence}. Thereafter, we abstract it to the aggregation of Markov chains. Consider the dataset $\mathcal{X}=\{x_i\}$ illustrated in the Figure \ref{fig: IntroFig}(a). It is apparent from the figure that the dataset comprises of $3$ groups of similar datapoints, i.e., 3 {\em natural} clusters. Figure \ref{fig: IntroFig}(b), \ref{fig: IntroFig}(c), and \ref{fig: IntroFig}(d) illustrate the corresponding clustering solutions obtained with $1$, $2$, and $3$ number of clusters, respectively. $\bar{T}_1$, $\bar{T}_2$, and $\bar{T}_3$ denote the {\em variance} (or, spread of data points) within the largest cluster in the respective clustering solutions. 

Note that there is a significant reduction in the variance of the largest cluster from $\bar{T}_2$ in Figure \ref{fig: IntroFig}(c) (with 2 clusters) to $\bar{T}_3$ in Figure \ref{fig: IntroFig}(d) (with 3 clusters). Contrastingly, the variance $\bar{T}_1$ in Figure \ref{fig: IntroFig}(b) (with 1 cluster) is comparable to the variance $\bar{T}_2$ in Figure \ref{fig: IntroFig}(c) (with 2 clusters). The significant reduction in the variance from $\bar{T}_2$ to $\bar{T}_3$ coincides with the progression from the clustering solution with 2 clusters (in Figure \ref{fig: IntroFig}(c)) to the clustering solution with 3 clusters (in Figure \ref{fig: IntroFig}(d)) - where $3$ is the {\em natural} number of clusters in Figure \ref{fig: IntroFig}(a). The work done in \cite{srivastava2019persistence} exploits the above observation. It proposes that the clustering solution with {\em true} number of clusters exhibits {\em a significant drop in the variance of the largest cluster} when compared to the clustering solution with one less cluster. \cite{srivastava2019persistence} refers this to as the clustering solution with true number of clusters being more {\em persistent} than other clustering solutions. The simulations in \cite{srivastava2019persistence} demonstrate that the above idea significantly outperforms popular benchmark methods \cite{tibshirani2001estimating,feng2007pg,pelleg2000x,sugar2003finding,hamerly2004learning,kalogeratos2012dip} in literature when applied to diverse synthetic and standard datasets.

We now extend the above idea to the aggregation of Markov chains. Analogous to the {\em variance} of a cluster in \cite{srivastava2019persistence} we establish the {\em heterogeneity} of a superstate in the context of Markov chains. Subsequently, we define the {\em marginal return} of an aggregated model as the {\em reduction in the heterogeneity} of the largest superstate in comparison to the aggregated model with one less superstate. We illustrate the idea further as follows. Consider a set of $K$ aggregated models each with different number $k\in\{1,\hdots,K\}$ of superstates. Analogous to the above case of data clustering we propose that {\em the aggregated model with the largest marginal return estimates the true number $k_t$ of superstates}. The idea is that there is a significant drop in the heterogeneity of the largest superstate in the aggregated models $\bar{M}_{k_t-1}$ with $k_t-1$ superstate in comparison to the aggregated model $\bar{M}_{k_t}$ with $k_t$ superstates. This could be owing to several reasons. For instance, the largest superstate in $\bar{M}_{k_t-1}$ combines {\em distinct} groups of {\em similar} states in the original Markov chain resulting into much larger heterogeneity, whereas the aggregated model $\bar{M}_{k_t}$ possibly identifies each group of {\em similar} states distinctly thus exhibiting much lesser heterogeneity. As in the case of data-clustering, the simulations on multiple synthetic and real-world Markov chains demonstrate the efficacy of the above proposed methodology. 

{\bf Quantifying heterogeneity and marginal return: } The quantification of persistence of clustering solutions and the subsequent characterization of natural number of clusters in \cite{srivastava2019persistence} is motivated from the Deterministic Annealing (DA) algorithm presented in \cite{rose1998deterministic}. Deterministic Annealing (DA) is a Maximum Entropy Principle \cite{jaynes2003probability} based clustering algorithm. It operates by determining a single cluster at the beginning of the algorithm. As the annealing proceeds the number of distinct clusters increases at specific instances referred to as {\em phase transitions} \cite{rose1998deterministic}. The notion of persistence and the subsequent characterization of natural clusters in \cite{srivastava2019persistence} are motivated by the analytical conditions determining phase transitions in DA. 

The work done in \cite{xu2014aggregation} presents a deterministic annealing (DA) approach to the aggregation of Markov chains. Here, the problem formulation and the associated cost functions are fundamentally distinct from the data clustering problem in \cite{rose1998deterministic}; however, the aggregation algorithm presented in \cite{xu2014aggregation} undergoes similar phase transition phenomenon. In particular, the algorithm begins with determining an aggregated model with a single superstate. As the annealing proceeds, the number of superstates in the model increase at phase transitions. In this paper we explicitly determine the analytical condition on the phase transition phenomenon occurring in \cite{xu2014aggregation}. Thereafter, analogous to the above case of clustering we quantify our notion of marginal return and the heterogeneity of a superstate based on these analytical conditions. Further, our characterization of the best choice for the number of superstates in the aggregated model is also motivated from the phase transition phenomenon. We substantiate on the exact expressions of heterogeneity and marginal return in the later sections.

Our simulations demonstrate the efficacy of marginal return in estimating the true number $k_t$ of superstates. We observe that marginal return value $\nu(k_t)$ is as high as $50$ to $400$ times the second largest value for Markov chains where the true number $k_t$ is discernible from the heatmap of the associated transition matrices $\Pi$. For Markov chains where true number $k_t$ is not discernible from the heatmaps, i.e. the underlying true number $k_t$ of superstates is not visually apparent, the persistence value $\nu(k_t)$ is still the largest and as high as $1.2$ to $5$ times the second largest value. In our simulations on Markov chains that model real-life scenarios - (a) emotion transitions in brain networks, and (b) letter bigram dataset, we observe that marginal return provides meaningful insights. In particular, it captures the distinct types of emotions (positive and negative) in the former, and the different types of English alphabets (vowels and consonants) in latter. 

This paper is organized as follows. Section \ref{sec: Agg_MarkovChain} summarizes the MEP-based aggregation framework in \cite{xu2014aggregation}, and derives analytical conditions for phase transition phenomenon. Section \ref{sec: MainResults} presents the main definitions, and the results of the paper. In Section \ref{sec: Simulations} we present simulations on synthetic and real-world Markov chains. Section \ref{sec: AnalysisDisc} provides a general analysis of our proposed method.

\section{Aggregation of a Markov Chain}\label{sec: Agg_MarkovChain}
In this section we briefly illustrate the Markov chain aggregation framework adopted in \cite{xu2014aggregation}. As stated above this framework forms the foundation for our quantification of heterogeneity of a superstate, and the marginal return of an aggregated model. Consider a Markov chain $(X,\Pi)$ with state space $X=\{x_i:1\leq i \leq N\}$, and the transition probability matrix $\Pi = (\pi_{ij})\in\mathbb{R}^{N\times N}$. The objective is to determine an aggregated representative chain $(Y,\Psi,\Phi)$ with $M$ $\ll N$ (super-) states, where $Y=\{y_j:1\leq j\leq M\}$ denotes the state space, $\Psi=(\psi_{jk})\in\mathbb{R}^{M\times M}$ denotes the transition probability matrix, and $\Phi:X\rightarrow Y$ is the associated {\em partition} function such that the state $x_i\in X$ is represented by the superstate $\Phi(x_i)\in Y$.

The framework associates a distribution vector $z(j)=(z_{j1},\hdots,z_{jN})\in\mathbb{R}^N$ to each superstate $y_j\in Y$. The distribution vector $z(j)$ captures the relation of the superstate $y_j$ to all the $N$ states in $X$. In fact, we can interpret $z_{jk}$ as the probability of transition from the superstate $y_j$ to the state $x_k$. Please refer to Figure \ref{fig: FrameworkXu} for a graphical interpretation of $z(j)$. Without loss of generality $\sum_{k=1}^Nz_{jk}=1$ and $z_{jk}\geq 0$ for all $k$. 

In the original Markov chain $(X,\Pi)$ the transition probability vector $\pi(i)=(\pi_{i1},\hdots,\pi_{iN})$ captures the relation of the state $x_i$ to all the $N$ states in $X$. Thus, the distance $d(x_i,y_j)$ between the state $x_i$ and the superstate $y_j$ is given by the {\em relative entropy} between the transition probability vector $\pi(i)$ and the distribution vector $z(j)$, that is, 
\begin{align}\label{eq: distMetric}
\text{\small$d(x_i,y_j) = \sum_{k=1}^N \pi_{ik}\log\frac{\pi_{ik}}{z_{jk}}~~ \forall~~1\leq i\leq N, 1\leq j\leq M.$}
\end{align}
\begin{figure}
    \vspace{1em}
    \centering
    \includegraphics[width=0.64\columnwidth]{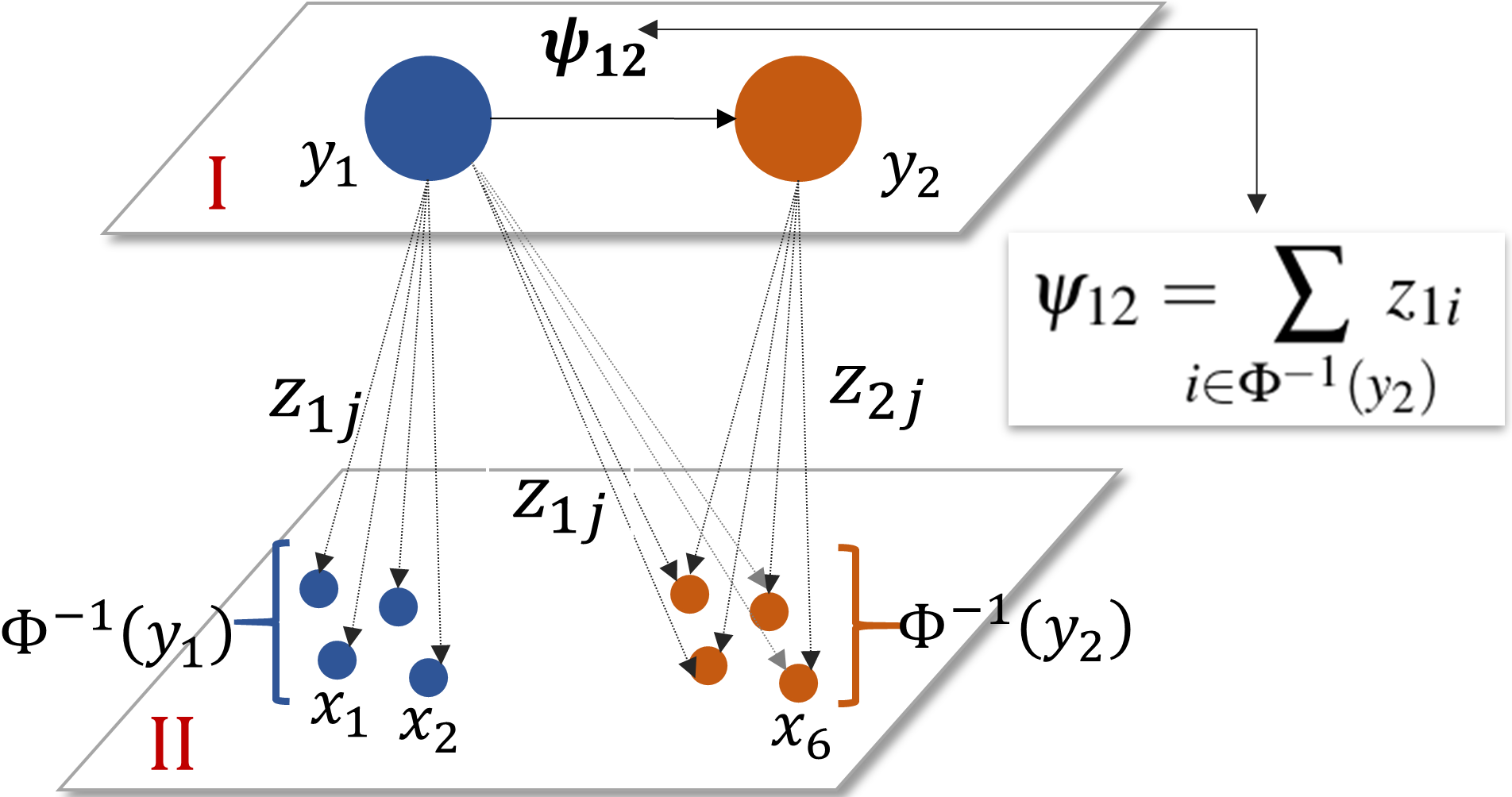}
    \caption{The plane $\mathrm{I}$ represents the aggregated model for the Markov chain in the plane $\mathrm{II}$. Each superstate $y_j$ in $\mathrm{I}$ is associated to each state $x_i$ in $\mathrm{I}\mathrm{I}$ via the weight $z_{ji}$. The set of states represented by superstate $y_j$ is given by $\Phi^{-1}(y_j)$.}\label{fig: FrameworkXu}\vspace{-1em}
\end{figure}

Thereafter, \cite{xu2014aggregation} determines a set of $M$ distribution vectors $\{z(j)\}_{j=1}^M$ such that the cumulative distance of each state $x_i$ to its closest superstate $y_j:=\Phi(x_i)$ (i.e., the superstate that represents $x_i$) is minimized. In particular, \cite{xu2014aggregation} addresses the following optimization problem
\begin{align}\label{eq: Distortion}
\begin{split}
\text{\small$\min_{Z,\Phi}$}&\text{\small$\quad D(\Pi,Z,\Phi) = \sum_{i=1}^N \rho_i~ d(x_i,\Phi(x_i))$}\\
\text{\small subject to}&\text{\small$\quad z_{jk}\geq 0~\forall~j,k~\text{and}~ z(j)^{\top}\mathbf{1}_N = 1~\forall ~j,$}
\end{split}
\end{align}
where $\rho_i$ denotes a given relative weight of the state $x_i$. Subsequently, \cite{xu2014aggregation} determines the transition matrix $\Psi=(\psi_{jk})$ of the aggregated chain in terms of the distribution vectors $\{z(j)\}_{j=1}^M$. In particular, the transition probability $\psi_{jk}$ from the superstate $y_j$ to $y_k$ is the cumulative sum of the weights $z_{ji}$'s from superstate $y_j$ to all states $x_i$ that are represented by the superstate $y_k$ (see Figure \ref{fig: FrameworkXu}), i.e.,
\begin{align}
\text{\small$\psi_{jk}=\sum_{\scriptstyle x_i\in\Phi^{-1}(y_k)}z_{ji}~\forall~1\leq j,k\leq M.$}
\end{align}

The partition function $\Phi:X\rightarrow Y$ in (\ref{eq: Distortion}) associates the state $x_i$ to a {\em particular} superstate $y_j=\Phi(x_i)\in Y$. In its solution methodology, \cite{xu2014aggregation} replaces this partition function by the {\em soft} partition weights $p_{j|i}\in[0,1]$ that determine the association of the state $x_i$ to different superstates $\{y_j\in Y\}$. Without loss of generality $\sum_{j}p_{j|i}=1$ and $p_{j|i}\geq 0$ for all $i$, $j$. Thereafter, Maximum Entropy Principle (MEP) is utilized to design the partition weights $\{p_{j|i}\}$. In particular, \cite{xu2014aggregation} determines the most {\em unbiased} partition weights $\{p_{j|i}\}$ and the distribution vectors $\{z(j)\}$ that maximize the corresponding Shannon entropy $H$ such that the expected cost function $\mathbb{E}_p[D]$ attains a pre-specified value $d_0$, i.e, the framework solves the following optimization problem
\begin{align}
\begin{split}
\text{\small$\max_{\{p_{j|i}\},\{z(j)\}}$}&\quad \text{\small$H = -\sum_{i=1}^N \rho_i \sum_{j=1}^M p_{j|i}\log p_{j|i}$}\label{eq: MEP_Optim}\\
\text{\small subject to}&\quad\text{\small$\mathbb{E}_p[D]:= \sum_{i=1}^N \rho_i \sum_{j=1}^M p_{j|i} d(x_i,y_j)= d_0$}\\
&\quad \text{\small$z_{jk}\geq 0,~z(j)\mathbf{1}_N=1~\forall j,k.$}
\end{split}
\end{align}
Minimizing (local) the Lagrangian $\mathcal{L}=(\mathbb{E}_p[D]-d_0)-TH$ (where $T$ is the Lagrange parameter) of above optimization problem with respect to $\{p_{j|i}\}$ and $\{z(j)\}$ results into 
\begin{align}
& p_{j|i} = \frac{\exp\{-(1/T) d(x_i,y_j)\}}{\sum_k \exp\{-(1/T) d(x_i,y_k)\}},\quad Z = P^{t}\Pi\label{eq: Soft_Association}\\
&\text{where } [P]_{ij} = \frac{\rho_i p_{j|i}}{\sum_t \rho_t p_{j|t}} \text{ and  }Z=[z(1),\hdots,z(M)]^{\top}\label{eq: RepDist_Association}.
\end{align}
The resulting Lagrangian $\mathcal{L}(Z)$ is given by
\begin{align}\label{eq: ApproxDistortion}
\text{\small$\mathcal{L}(Z) = -T\sum_{i=1}^N \rho_i \log\sum_{j=1}^M \exp\Big\{-\frac{1}{T} \sum_{k=1}^N \pi_{ik}\log \frac{\pi_{ik}}{z_{jk}}\Big\}$}.
\end{align}
It is known from the sensitivity analysis \cite{jaynes2003probability} that the large (small) values of $d_0$ in (\ref{eq: MEP_Optim}) is analogous to large (small) values of the Lagrange parameter $T$. In fact, the optimization problem in (\ref{eq: MEP_Optim}) is repeatedly solved at decreasing values of $T$ (or equivalently, decreasing values of $d_0$). At large values of $T\rightarrow \infty$, the Lagrangian $\mathcal{L}$ is dominated by the convex function $-H$, the partition weights $\{p_{j|i}\}$ in (\ref{eq: Soft_Association}) are uniformly distributed ($p_{j|i}=1/M$), and all the distributions $\{z(j)\}$ in (\ref{eq: Soft_Association}) are {\em co-incident}, i.e., effectively {\em one} distinct superstate is obtained at large values of $T$. As $T$ decreases further the Lagrangian is more and more dominated by the cost function $\mathbb{E}_p[D]$, the partition weights $\{p_{j|i}\}$ are no longer uniform, and the coincident distributions $\{z(j)\}$ split into distinct groups, i.e., the effective number of distinct superstates increases. 

In this article we exploit this heirarchical splitting to define the notions of heterogeneity and the marginal return. Insights from the simulations of the algorithm in \cite{xu2014aggregation} demonstrate that there exist certain critical temperature $T=T_{\text{cr}}$ values at which the solution undergoes the phenomenon of {\em phase transition}. This phenomenon is characterized by the increase in the number of distinct distributions in $\{z(j)\}$; or, equivalently increase in the number of distinct superstates. These critical temperatures $T_{\text{cr}}$'s occur when the critical point $Z^*$ of the Lagrangian $\mathcal{L}$, given by $\frac{\partial \mathcal{L}(Z^*)}{\partial Z}=0$ is no longer the local minima, that is, when for at least one perturbation direction  $\Psi\in\mathbb{R}^{M\times N}$, the Hessian $\mathcal{H}(Z^*,P^*,\Psi,T):=$
\begin{align}\label{eq: Hessian}
\text{\small$\lim_{\epsilon\rightarrow 0}\frac{\partial^2 \mathcal{L}(Z^*+\epsilon\Psi)}{\partial \epsilon^2}\hspace{0ex}$=}&\text{\small$\sum_{j=1}^M q_j \psi_j^{\top} (\Lambda_{T}(j) - \frac{1}{T} C_{T}(j))\psi_j$} \nonumber\\ \text{\small$+$}&\text{\small$T\sum_{i=1}^N\Big(\sum_{j=1}^M p_{j|i}\Big[\frac{\pi(i).}{z^*(j)}\Big]^{\top}\psi_j\Big)^2,$}
\end{align}
is no longer positive. Here $q_j = \sum_{i=1}^N\rho_ip_{j|i}$, $\Psi=[\psi_1,\hdots,\psi_M]^{\top}\in\mathbb{R}^{M\times N}$ is perturbation direction,
\begin{align}\label{eq: imp_matrices}
\begin{split}
\text{\small$\Lambda_{T}(j)$}&\text{\small$= \text{diag}\Bigg\{\frac{\sum_{i=1}^Np_{i|j}\pi(i).}{z^*(j).^2}\Bigg\},$}\\
\text{\small$C_{T}(j)$}&\text{\small$= \sum_{i=1}^Np_{i|j}\Big[\frac{\pi(i)-z^*(j).}{z^*(j)}\Big]\Big[\frac{\pi(i)-z^*(j).}{z^*(j)}\Big]^{\top},$}
\end{split}
\end{align}
and $p_{i|j}=(\rho_ip_{j|i}/\sum_t \rho_tp_{j|t})$ is the posterior distribution. We derive the expression for critical temperature $T_{cr}$ as below.

\begin{theorem}
The value of critical temperature $T_{cr}$ at which the the Hessian $\mathcal{H}(Z^*,P^*,\Psi,T)$ in (\ref{eq: Hessian}) is no longer positive for some perturbation direction $\Psi$, and $Z^*$ in (\ref{eq: Soft_Association}) undergoes phase transition is given by
\begin{align}
&\text{\small$T_{\text{cr}}$} := \text{\small$\max_{1\leq j\leq M}\Big[T_{cr,j}\Big]$},\text{ where }  \text{\small$T_{cr,j}=\lambda_{\max}\big(\mathcal{C}_T(j)\big)$, and}\label{eq: cric_beta}\\
&\text{\small $\mathcal{C}_T(j)$} = \text{\small$\sum_{i=1}^N[P]_{ij}\Big[\Theta^{\top}\frac{\pi(i)-z^*(j).}{z^*(j)}\Big]\Big[\Theta^{\top}\frac{\pi(i)-z^*(j).}{z^*(j)}\Big]^{\top}$,}\label{eq: cric_beta2}
\end{align}
$\lambda_{\max}(\cdot)$ is the largest eigenvalue, $\frac{a.}{b}$ denotes element-wise division of vectors $a$ and $b$, and  $\Theta\in\mathbb{R}^{N\times N-1}$ corresponds to the constraint $z(j)^{\top}\mathbf{1}_N=1$ $\forall$ $1\leq j\leq M$ in (\ref{eq: Distortion}) and (\ref{eq: MEP_Optim}). See \ref{app: First} for further details.
\end{theorem}
\begin{proof}
Please refer to the \ref{app: First}.
\end{proof}

{\em Re-interpreting $T_{cr}$, $T_{cr,j}$, and $\mathcal{C}_T(j)$: } As stated earlier, the annealing temperature $T:=T_{cr}$ (computed in (\ref{eq: cric_beta})-(\ref{eq: cric_beta2})) marks the increase in the number of distinct superstates in the aggregated chain, i.e., it characterizes the phase transitions in the DA-based algorithm \cite{xu2014aggregation}. However, the expressions in (\ref{eq: cric_beta})-(\ref{eq: cric_beta2}) are also further interpretable beyond the current context of phase transitions. For instance, $\mathcal{C}_T(j)$ in (\ref{eq: cric_beta2}) is a {\em soft co-variance matrix} of the posterior distribution $[P]_{ij}$ corresponding to the superstate $y_j\in Y$, $T_{cr,j}$ in (\ref{eq: cric_beta}) is the maximum eigenvalue of $\mathcal{C}_T(j)$, and thus, captures the {\em maximum variance} within the transition probabilities $\{\pi(i)\}$ of the states $\{x_i\}$ represented by the superstate $y_j$. In other words, $T_{cr,j}$ can be interpreted as a {\em measure of heterogeneity} within the states represented by the superstate $y_j$, and $T_{cr}$ can be interpreted as the maximum heterogeneity among all the superstates $\{y_j\in Y\}$.

\begin{figure*}
\centering
\includegraphics[width=0.9\textwidth]{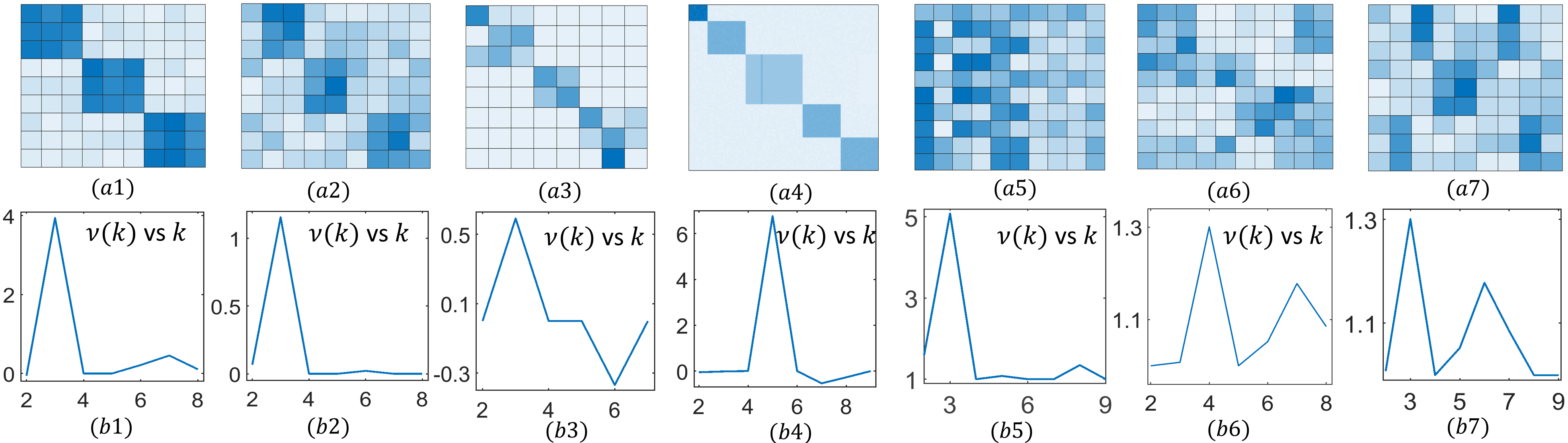}
\caption{Illustrates the efficacy of Algorithm \ref{alg: Algorithm1} in estimating $k_t$. (a1)-(a4) demonstrate the heatmaps for the transition matrices of NCD Markov chains generated such that $N=9$, $k_t=3$ in (a1)-(a2), $N=8$, $k_t=3$ in (a3), and $N=100$, $k_t=5$ in (a4). (b1)-(b4) are the corresponding persistence plots t  hat clearly indicate $\nu(k_t)>\nu(k)$ $\forall$ $k\neq k_t$. (a5)-(a7) illustrate Markov chains generated by considering multiple copies of random vectors $\{\xi_i\in\mathbb{R}^N\}_{i=1}^{k_t}$; where $N=10$, $k_t=3$ in (a5), $N=10$, $k_t=4$ in (a6), and $N=9$, $k_t=3$ in (a7). (b5)-b(7) are the corresponding persistence plot that accurately estimate $k_t$. }\label{fig: Figure3}\vspace{-0.4cm}
\end{figure*}

The soft co-variance matrix  $\mathcal{C}_T(j)$ in (\ref{eq: cric_beta2}) depends on the partition weights $[P]_{ij}$ and distribution vectors $Z=P^{\top}\Pi$ in (\ref{eq: RepDist_Association}) that result from DA-based aggregation algorithm in \cite{xu2014aggregation} $-$ making the above quantification $T_{cr,j}$ of heterogeneity dependent on the algorithm. Thus, in the following section, we adapt $\mathcal{C}_T(j)$ to incorporate aggregated chains irrespective of the algorithm used to determine them. In particular, for a given aggregated representation $(Y,\Psi,\Phi)$ of a Markov chain $(X,\Pi)$, we replace the soft partitions $[P]_{ij}$ in $\mathcal{C}_T(j)$ with the hard-partitions prescribed by the partition function $\Phi:X\rightarrow Y$. Subsequently, we provide the expression for {\em marginal return}, and the Algorithm \ref{alg: Algorithm1} that estimates a choice for the number of superstates to be considered in the aggregated chain.

\section{Covariance Matrix and Marginal Return}\label{sec: MainResults}
Consider a Markov chain $(X,\Pi)$ with state space $X=\{x_i:1\leq i \leq N\}$, and the transition probability matrix $\Pi = (\pi_{ij})\in\mathbb{R}^{N\times N}$, its aggregated representative chain $(Y,\Psi,\Phi)$ with state space $Y=\{y_j:1\leq j\leq M\}$, the transition probability matrix $\Psi=(\psi_{jk})\in\mathbb{R}^{M\times M}$, where $M\ll N$, and the associated partition function $\Phi:X\rightarrow Y$ such that the state $x_i\in X$ is represented by the superstate $y_j:=\Phi(x_i)\in Y$.
\begin{definition}
The {\em co-variance matrix} $C_{X}^{\Phi}(j)$ corresponding to the superstate $y_j\in Y$ in the aggregated chain is given by
\begin{align}
&\text{\small$C_{X}^{\Phi}(j)$} := \text{\small$\mathlarger{\sum}_{i=1}^N [Q]_{ij}\Big[\Theta^{\top} \frac{(\pi(i)-w(j)).}{w(j)}\Big]\Big[\Theta^{\top} \frac{(\pi(i)-w(j)).}{w(j)}\Big]^{\top}$,}\label{eq: Cov_Matrix}\\
&\text{\small where } \text{\small$[Q]_{ij}=
\begin{cases}
1, & \text{if } x_i\in \Phi^{-1}(y_j)\\
0, & \text{otherwise}
\end{cases},\quad Q := [Q]_{ij},$}\label{eq: Cov_Matrix2}
\end{align}
$W=[w(1),\hdots,w(M)]^{\top}=Q^{\top}\Pi$ denotes the distribution vector as obtained in (\ref{eq: Soft_Association}), $\Theta\in\mathbb{R}^{N\times N-1}$ corresponds to the constraint $z(j)^{\top}\mathbf{1}_N=1$ $\forall$ $1\leq j\leq M$ in (\ref{eq: Distortion}), $\frac{a.}{b}$ denotes element-wise division of vectors $a$ and $b$, and $[\cdot]^{\top}$ denotes the transpose.
\end{definition}
\begin{definition}
The {\em marginal return} $\nu(k)$ of the aggregated model with $k$ number of superstates among the given $K$ aggregated representations $\{(Y_k,\Psi_k,\Phi_k):|Y_k|=k\}_{k=1}^K$ of a Markov chain $(X,\Pi)$ is given by
\begin{align}
&\text{\small$\nu(k) := \log \bar{T}_{k-1} - \log \bar{T}_{k}, \text{ where }$}\label{eq: Persistence}\\
&\text{\small$\bar{T}_{l} := \max_{1\leq j \leq l}\big[\bar{T}_{l,j}\big]\text{ for } l\in\{k-1,k\}, \bar{T}_{l,j}=\lambda_{\max}(C_{X}^{\Phi_l}(y_j))$}\label{eq: Persistence2}
\end{align}
denotes the heterogeneity of superstate $y_j$, and $\lambda_{\max}(\cdot)$ denotes the largest eigenvalue.
\end{definition}

As illustrated in the Section \ref{sec: Introduction}, the aggregated model with largest marginal return estimates the true number $k_t$ of superstate underlying a Markov chain, i.e.
\begin{align}
k_t:=\arg\max_{1\leq k\leq K}\nu(k),
\end{align}

The following algorithm computes the marginal return of the $K$ aggregated representations $\{(Y_k,\Psi_k,\Phi_k):|Y_k|=k\}_{k=1}^K$ of a Markov chain $(X,\Pi)$, and estimates the corresponding true number $k_t$ of superstates.

\begin{algorithm}
\textbf{Input: $(X,\Pi)$,  $\{(Y_k,\Psi_k,\Phi_k):|Y_k|=k\}_{k=1}^K$};  
\textbf{Output: }{$\nu(k)$, $k_t$}\\
\For {$k=1$ to $K$}{
From partition $\Phi_k:X\rightarrow Y$ determine $Q$ in (\ref{eq: Cov_Matrix2})\\
Compute $W:=[\cdots w(j)\cdots]^{\top}=Q^{\top}\Pi$ and $\bar{T}_k$ using (\ref{eq: Persistence2}). 
}
compute $\nu(k)$ in (\ref{eq: Persistence}) $\forall$ $2\leq k\leq K$, and $k_t:=\arg\max_k \nu(k)$.
\caption{Marginal return and Number of superstates}\label{alg: Algorithm1}
\end{algorithm}\vspace{-0.4cm}
\section{Simulations}\label{sec: Simulations}
In this section we demonstrate the efficacy of marginal return in comparing different aggregated models, and estimating true number $k_t$ of superstates underlying the given Markov chain. We use the Algorithm \ref{alg: Algorithm1} that takes in the aggregated models at different number $k$ of superstates as inputs, and outputs the corresponding marginal return and estimated $k_t$. To demonstrate the generality of our proposed method, we use aggregated models obtained from different aggregation algorithms. We use the algorithm presented in \cite{geiger2014markov} on the first four example simulations in Figure \ref{fig: Figure3}, and \cite{xu2014aggregation} on the remaining examples in Figures \ref{fig: Figure3} and \ref{fig: Fig4}. 

{\em NCD Markov Chains:} The transition matrix $\Pi$ for the NCD Markov chains \cite{ando1963near} presumes the following structure $\Pi = \Pi^* + \epsilon C$, where $\Pi^*$ is a block diagonal matrix and $C$ adds a perturbation of the range $\epsilon$. Naturally, the number $k_t$ of superstates for such Markov chains can be approximated to be the number of block diagonals in $\Pi^*$. Figures \ref{fig: Figure3}(a1)-a(2) illustrate the heatmaps of the transition probability matrix for two such Markov chains obtained at different levels of perturbation to a three block diagonal matrix $\Pi^*$. Figures \ref{fig: Figure3}(b1)-(b2) illustrate the corresponding persistence $\nu(k)$ versus $k$ plots which correctly identifies the true number $k_t=3$ of superstates in both cases. Note that, even though the transition matrix $\Pi$ in Figure \ref{fig: Figure3}(a2) is highly perturbed, our method accurately identifies number of superstates thereby, demonstrating the robustness of the marginal return $\nu(k)$. 

Figure \ref{fig: Figure3}(a3) illustrates a specialized NCD Markov chain with Courtois Transition Matrices. These Markov chains are well studied in literature for their slow convergence to the steady state \cite{elsayad2002numerical}. As is evident from the above heatmap in Figure \ref{fig: Figure3}(a3), the transition matrix comprises of $3$ dominant block diagonal. Thus, the corresponding Markov chain comprises of $k_t=3$ number of superstates. This is captured in our marginal return analysis in the Figure \ref{fig: Figure3}(b3). Figure \ref{fig: Figure3}(a4) demonstrates the state transition matrix of a large ($N=100$ states) NCD Markov chain, considered in \cite{xu2015clustering}, constituting $5$ underlying block diagonals (one each of size $10\times 10$ and $30\times 30$, and three of size $20\times 20$). Figure \ref{fig: Figure3}(b4) confirms largest marginal return at $k=5$ number of superstates, and thus, appropriately estimates $k_t$.
\begin{comment}
Further, under additional constraints such as on the number of superstates permitted ($c_0\leq k\leq c_1$) in aggregated chains, where $k_t$ superstates do not satisfy the constraint, one can utilize the marginal return plot to determine the next best choice for the size of the aggregated model. For instance, Figure \ref{fig: Figure3}(b7) indicates that second largest marginal return at $k=6$ for the Markov chain in Figure \ref{fig: Figure3}(a2) and thus, estimates the appropriate size of the aggregated chain when models with the $k\leq 3$ are not permitted.
\end{comment}
{\em Randomly generated Markov Chains:} In the following simulations we consider $N$ state Markov chains where the transition matrices $\Pi=[\xi_{i_1},\hdots,\xi_{i_N}]^{\top}+\epsilon C$ are generated from $k_t$ random distribution vectors  $\{\xi_i\}_{i=1}^{k_t}$, and $\epsilon C$ introduces random perturbations. Naturally, $k_t$ estimates true number of superstates underlying the above Markov chains. Figure \ref{fig: Figure3}(a5) illustrates one such scenario where $\Pi\in\mathbb{R}^{10\times 10}$ is generated by perturbing multiple copies of $\{\xi_i\}_{i=1}^3$, and the marginal return analysis in Figure \ref{fig: Figure3}(b5) accurately determines $k_t=3$ as the estimate for the true number of superstates. Similarly, Figures \ref{fig: Figure3}(a6) and \ref{fig: Figure3}(a7) illustrate the transition matrices $\Pi\in\mathbb{R}^{10\times 10}$ generated from $4$ and $3$ random distribution vectors, respectively, and the marginal return plots in Figure \ref{fig: Figure3}(b6) and \ref{fig: Figure3}(b7) accurately estimate the underlying true number of superstates. Note that even though the perturbations are significantly large in Figures \ref{fig: Figure3}(a5)-(a7), marginal return $\nu(k)$ accurately estimates the number $k_t$ of superstates underlying the original Markov chains.

{\em Emotion Transitions in Brain Network: } We consider the example Markov chain that models the transition between different emotional states of a person \cite{thornton2017mental}. The Figure \ref{fig: Fig4}(a1) illustrates the transition matrix $\Pi$ comprising $N=22$ states - each corresponding to an emotion in the set {\em$X=\{$anxious, jittery, irritable, vigorous, alert, lively	happy, attentive, intense, full-of-pep, excited, distressed,	strong, stirred-up, nervous, upset, touchy, bold, temperamental, quiet, talkative, insecure$\}$}. We obtain different aggregated representations of $\Pi$ using the algorithm presented in \cite{xu2014aggregation}, and observe that the marginal return $\nu(k)$ is the largest for $k_t=2$ (see Figure \ref{fig: Fig4}(b1)). This indicates the presence of two superstates underlying the original Markov chain in Figure \ref{fig: Fig4}(a1). We note that the largest marginal return $\nu(k)$ at $k_t=2$ is in accordance with the fact that the emotions in $X=X_p\sqcup X_n$ can be classified into positive $X_p$ and negative $X_n$ emotions, where $X_p=\{$vigorous, alert, lively, happy, attentive, intense, full-of-pep, excited, strong, bold, quiet$\}$, and $X_n=\{$anxious, jittery, irritable, vigorous, distressed, stirred-up, nervous, upset, touchy, temperamental, insecure, talkative$\}$. Thus, our notion of marginal return captures the number of distinct type of emotions in $X$. Similar analysis can be done on different brain network models \cite{vidaurre2017brain, vidaurre2018spontaneous} to qualify, for instance, the network states responsible for either {\em sensory-motor systems}, or for {\em higher-order cognition} (such as language).

\begin{figure}
\centering
\includegraphics[width=0.65\columnwidth]{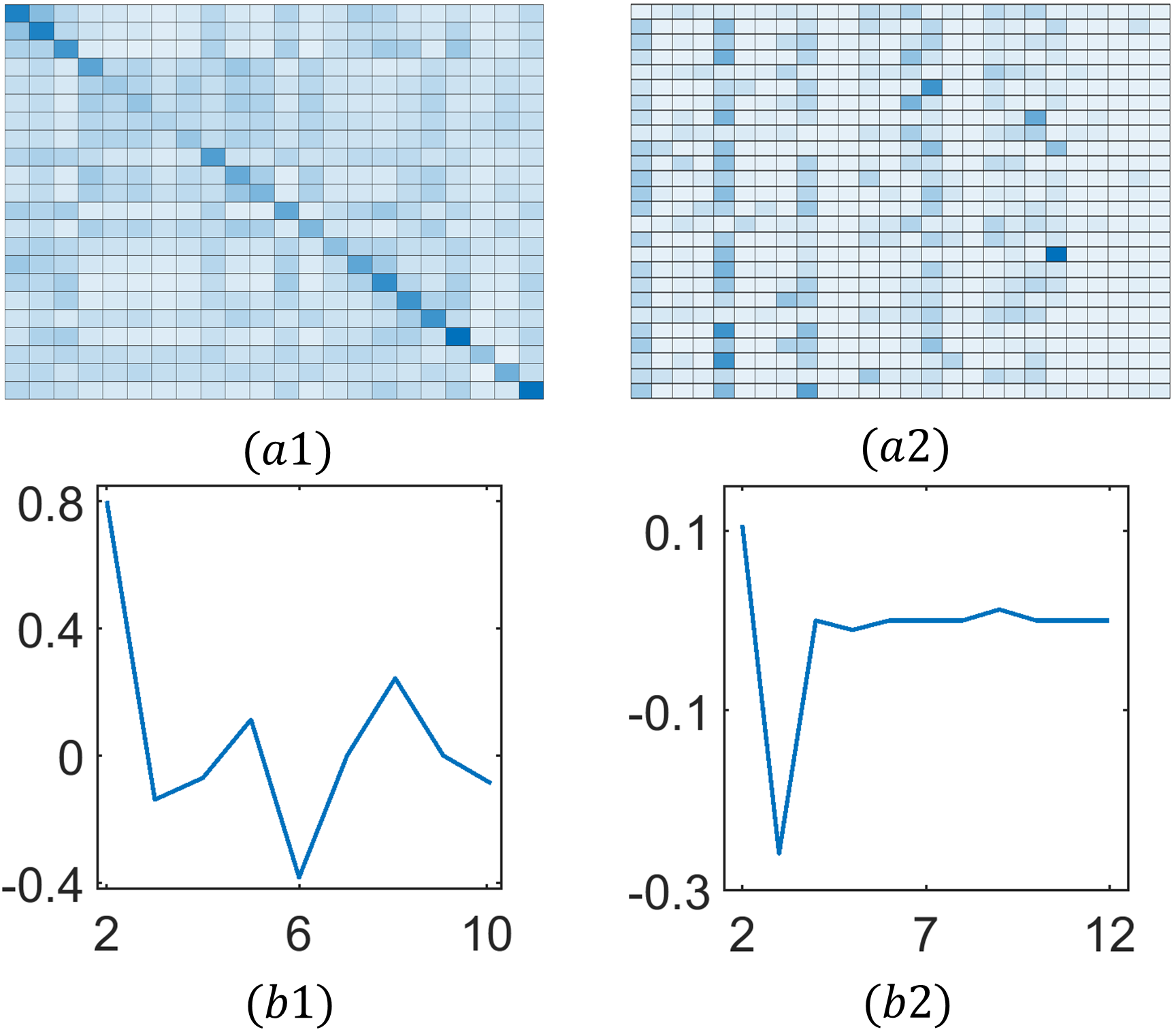}
\caption{Marginal return analysis in realistic Markov chains. (a1) Transition matrix for brain emotion transitions where each state corresponds to an emotion. (a2) Transition matrix for letter bigram dataset. Each state correspond to an English alphabet. (b1) Maximum $\nu(k)$ observed at $k_t=2$ suggesting two types of underlying emotions. (b2) Largest $\nu(k)$ at $k_t=2$ indicating two types of alphabets in English language - vowels and consonants.}\label{fig: Fig4}
\end{figure}
\begin{table}
\centering
\begin{tabular}{|c|c|}
\hline
{\small$k$} & {\small State Partitioning}\\
\hline
{\small$2$} &  {\small$\{aeiouy\},\{bcdfghjklmnpqrstvwxz\}$}\\
{\small$3$} & {\small$\{aeiouy\},\{cdghknrstvwxz\},\{bfjlmpq\}$}\\
{\small$4$} & {\small$\{aeiouy\},\{hnrsvxz\},\{cdgkwt\},\{bfjlmpq\}$}\\
{\small$5$} & {\small$\{aeiouy\},\{bflmp\},\{cdgktw\},\{hnrlvxz\},\{jq\}$}\\\hline
\end{tabular}
\caption{State aggregations of letter bi-gram data obtained at different number of superstates. Obtained via aggregation algorithm in \cite{xu2014aggregation}.}\label{Tab: Agg}
\vspace{-0.4cm}
\end{table}
{\em Letter Bigram dataset:} We finally investigate an example from Natural Language Processing. We consider the letter bi-gram dataset \cite{norvig2013english} that enumerates the number of transitions from one alphabet to another based on the Google Corpus Data (collection of $97,565$ distinct words, which were mentioned $743,842,922,321$ times). We represent the entire dataset as a Markov chain where each state represents an individual alphabet. The transition matrix $\Pi\in\mathbb{R}^{26\times 26}$ (see Figure \ref{fig: Figure3}(a9)) of the Markov chain captures the frequency of transition from one alphabet (state) to another. Our idea is to demonstrate that {\em marginal return} captures a ``meaningful" number of superstates to be considered in the aggregated representation of the above Markov chain. Table \ref{Tab: Agg} illustrates aggregated models obtained at different number of superstates. Since the alphabets are of two types, i.e., either vowels or consonants, one may find it reasonable to aggregate the original Markov chain into two superstates - one each for vowels and consonants. The above intuitive argument is reinforced by the largest value of marginal return $\nu(k)$ that is obtained at $k_t=2$ (see Figure \ref{fig: Figure3}(b9)). Note that the corresponding aggregated model at $k_t=2$ (as illustrated in Table \ref{Tab: Agg}) partitions the alphabets into either vowel sounds $\{aeiouy\}$, or the consonants (where $y$ is a (semi)-vowel).

\section{Conclusion}\label{sec: AnalysisDisc}
Model-reduction techniques usually rely on the pre-specified size of the reduced model. Thus, there arises a need to methodically estimate this size. In the context of Markov chain aggregation, we exploit the phase transitions in \cite{xu2014aggregation}, and devise the notions of heterogeneity and marginal return. We demonstrate that the proposed notions are insightful in comparing aggregated models of different sizes, and estimating the number of superstates underlying the original Markov chain. The ideas and methods presented in this paper can be extended to several related combinatorial model-reduction problems. For instance, the graph clustering problem \cite{xu2015clustering}, that requires aggregating the nodes of a large graph into {\em super-nodes} and determining the edge-connections between them, or the co-clustering problem \cite{dhillon2003information}, that aggregates a given matrix $\mathcal{X}\in\mathbb{R}^{N_1\times N_2}$ to a smaller representative matrix $\mathcal{Y}\in\mathbb{R}^{M_1\times M_2}$, pose the similar optimization problem as in (\ref{eq: Distortion}). The appropriate size of the representative graph in the former, or the size of the aggregated matrix $Y$ in the latter can be estimated using similar ideas as elucidated in our work.
\section*{Acknowledgments}
The authors would like to acknowledge Bernhard Geiger, Rana Ali Amjad, Clemens Bloechl for sharing their code for the aggregation algorithm proposed in \cite{amjad2019generalized}. 
\appendix
\section{}\label{app: First}
\noindent Perturbation $\Psi=[\psi_1,\hdots,\psi_M]\in \mathbb{R}^{M\times N}$. Let $\psi_j=\Phi K_j$, where $\Phi\in \mathbb{R}^{N\times N}$ and $\Phi \mathbf{1}_N=0$ for admissible perturbation of $z(j)$ (i.e. $(z(j)+\epsilon \psi_j)\mathbf{1}_N=1$ in (\ref{eq: Distortion})).
\begin{lemma}\label{lem: Lemma1}
The Hessian {\small$\mathcal{H}(Z^*,P^*,\Psi,T)$} in (\ref{eq: Hessian}) loses rank when $\text{det}\big[\Phi^{\top}\big(\Lambda_T(j)-\frac{1}{T}C_T(j)\big)\Phi\big]=0$ for some $1\leq j\leq M$, where $\Lambda_T(j)$ and $C_T(j)$ are as defined in (\ref{eq: imp_matrices}).
\end{lemma}
\begin{proof}
Motivated from \cite{rose1998deterministic}. $\mathcal{H}(Z^*,P^*,\Psi,T)$ is positive for all perturbation $\Psi$ if and only if its first part (see (\ref{eq: Hessian})) is positive. `If' part is straightforward. For the `only if' part we show that $\exists$ a $\Psi$ such that second term in $\mathcal{H}$ vanishes. Let $J=\{j_1,\hdots,j_r\}$ denote the superstates with distribution $z_{j_0}$. Select $\Psi$ such that $\psi_j=0~\forall~j\notin J$ and $\sum_{j\in J}\psi_j=0$. For this perturbation the second term vanishes and $\mathcal{H}$ depends only on its first term. Thus, we establish the `only if' part. Replacing the $\psi_j$ in (\ref{eq: Hessian}) with $\Phi K_j$ we obtain the condition stated in the Lemma.
\end{proof}
\begin{lemma}\label{lem: Lemma2}
The condition {$\text{det}\big[\Phi^{\top}\big(\Lambda_T(j)-\frac{1}{T}C_T(j)\big)\Phi\big]=0$} for some {$1\leq j\leq M$} is attained at temperature value $T=T_{cr}:=\lambda_{\max}\big(\mathcal{C}_T(j)\big)$, where $\mathcal{C}_T(j)$ is given in (\ref{eq: cric_beta2}).
\end{lemma}
\begin{proof}
Rank of $\Phi = N-1$. Let $\Upsilon\in\mathbb{R}^{N\times N-1}$ such that Range($\Upsilon$) = Range($\Phi$). Let $W_j\in\mathbb{R}^{N-1}$ such that $\Upsilon W_j = \Phi K_j$. From Lemma \ref{lem: Lemma1} $\exists$ $W_j$ $:$ $W_j\Upsilon^{\top}\big(\Lambda_T(j)-\frac{1}{T}C_T(j)\big)\Upsilon W_j=0$. Let $H_0=\Upsilon^{\top}\Lambda_T(j)\Upsilon$, $H_1=\Upsilon^{\top} C_T(j)\Upsilon$. From Theorem 12.19 in \cite{laub2005matrix} $\exists$ $G\in\mathbb{R}^{N\times N-1}$ s.t. $G^{\top}H_0G=I$ and $G^{\top}H_1G=\mathcal{C}_T(j)$ where $G=L^{-t}P$, $LL^{\top}=H_0$, $P^{\top}[L^{-1}H_1L^-t]P=\mathcal{C}_T(j)$. Let $W_j=G\omega_j$, $\omega\in\mathbb{R}^{N-1}$, $\Rightarrow$ $w_j^{\top}(G^{\top}H_0G-\frac{1}{T}G^{\top}H_1G)\omega_j=0$ $\Rightarrow$ $\text{det}\big(I-\frac{1}{T}\mathcal{C}_T(j)\big)=0$ $\Rightarrow T_{cr} = \lambda_{\max}\big(\mathcal{C}_T(j)\big)$. Since $G^{\top}H_1G=\mathcal{C}_T(j)$ we have that $\Theta = \Upsilon G$.
\end{proof}

\bibliographystyle{IEEEtran}
\bibliography{IEEEabrv}

\end{document}